\documentclass[11pt]{IEEEtran}

\usepackage{amsmath,amsthm,amssymb,amsfonts,verbatim}
\usepackage{latexsym}
\usepackage{fullpage}

\usepackage[dvipsnames,usenames]{color}
\usepackage[pdftex,letterpaper=true,colorlinks=true,pdfpagemode=none,urlcolor=blue,linkcolor=blue,citecolor=BrickRed,pdfstartview=FitH]{hyperref}
\usepackage{hyperref}
\usepackage{times}

 \providecommand{\F}{\mathbb{F}}

\date{}

\newtheorem{lemma}{Lemma}
\newtheorem{theorem}{Theorem}

\newtheorem{defn}{Definition}

\def \mC {\mathcal{C}}

\def \mB {\mathcal{B}}
\def \mC {\mathcal{C}}

\def \Xi {{X^{[i]}}}

\def\rank {{\rm rank }}

\begin{document}

\onecolumn

\newpage
\title{On the List Decodability of Self-orthogonal\\ Rank Metric Codes}

\author{Shu Liu
\thanks{The author is with Division of Mathematical Sciences, School
of Physical and Mathematical Sciences, Nanyang Technological
University, Singapore 637371, Republic of Singapore
(email: SLIU017@ntu.edu.sg).}
% \thanks{The work was partially supported by the Singapore  Ministry of Education  Tier 1 grant 20/13.}
 }
\maketitle

\begin{abstract}
V. Guruswami and N. Resch prove that the list decodability of $\F_q$-linear rank metric codes is as good as that of random rank metric codes in~\cite{venkat2017}. Due to the potential applications of self-orthogonal rank metric codes, we focus on list decoding of them. In this paper, we prove that with high probability, an $\F_q$-linear self-orthogonal rank metric code over $\F_q^{n\times m}$ of rate $R=(1-\tau)(1-\frac{n}{m}\tau)-\epsilon$ is shown to be list decodable up to fractional radius $\tau\in(0,1)$ and small $\epsilon\in(0,1)$ with list size depending on $\tau$ and $q$ at most $O_{\tau, q}(\frac{1}{\epsilon})$. In addition, we show that an $\mathbb{F}_{q^m}$-linear self-orthogonal rank metric code of rate up to the Gilbert-Varshamov bound is $(\tau n, \exp(O_{\tau, q}(\frac{1}{\epsilon})))$-list decodable. 

\end{abstract}

\section{Introduction}
In the late 50's, P. Elias~\cite{P.E1957}, \cite{P.E1991} and J. M. Wozencraft~\cite{J.M.W1985} introduced list decoding.
Compared with unique decoding, list decoding can output a list of codewords which contains the correct transmitted codeword rather than output a unique codeword. Consequentially, the list size of list decoding can be bigger than $1.$ 

In coding theory, there exists a trade-off between the fraction of errors that can be corrected and the rate. The largest list size of the decoder's output is an important parameter in list decoding. In fact, we want to have a small list size. At least two reasons can be argued. The first reason is due to the usefulness of this list. After the output of this list, the next step is to utilize this list to decide what the original transmitted message is. This can be done, by outputting the codeword corresponding to the smallest error. If the list size is exponential, the decision step needs exponential time complexity. The second reason is that this list size provides us with a lower bound for the worst-case complexity of the decoding algorithm itself. So, if we require the decoding algorithm to be efficient, we need the list size to be as small as possible. 

\subsection*{List Decoding of Rank Metric Codes}
Rank metric codes is a set of $n\times m$ matrices over a finite field $\F_q.$ By the ring isomorphism between $\F_q^m$ and $\F_{q^m},$ an $n\times m$ matrix over $\F_q$ can be defined as a vector of length $n$ over the extension field $\F_{q^m}.$ Rank metric codes have been receivced many attention because of their applications in network coding~\cite{Sil2008}, \cite{Koe2008}, storage systems~\cite{Roth1991}, cryptography~\cite{E.M.A.V1991}, \cite{Ove2008}, and space-time coding~\cite{Lus2003}. 

Finding good list decodable rank metric codes attracts more and more researchers. In order to find the limit in which an efficient list decoding is possible, A. Wachter-Zeh provided lower and upper bounds on the list size in~\cite{Ant2013}, \cite{Zeh2012} and showed that the upper bound of the list size is exponential, for any decoding radius beyond half of the minimum distance. In addition, there exists an exponential list size rank metric code for any decoding radius is larger than half of the minimum distance. 
No efficient list decoding can be found for Gabidulin codes if decoding radius beyond the Johnson bound. 
Y. Ding~\cite{Ding2015} reveals that the Singleton bound is the list decoding barrier for any rank metric code.
%More precisely speaking, when the list decoding radius is beyond the Singleton bound, the list size is exponential in the length of the codes. 
%For the list decodability of random $\F_q$-linear rank metric codes, the Gilbert-Varshamov bound is the list decoding barrier. 
With high probability, the decoding radius and the rate of random rank metric codes satisfy the Gilbert-Varshamov bound with constant list size. In addition, with high probability, an $\F_q$-linear rank metric code can be list decoded with list decoding radius attaining the Gilbert-Varshamov bound with exponential list size. 
Since efficient list decoding radius of Gabidulin codes cannot be larger than the unique decoding radius, S. Liu, C. Xing and C. Yuan show that with high probability, a random subcode of a Gabidulin code is list decodable with decoding radius far beyond the unique decoding radius in~\cite{Liu2017}. However, for $\F_q$-linear rank metric codes, when the list decoding radius is beyond half of the minimum distance, the list size is exponential. V. Guruswami and N. Resch decrease the list size of $\F_q$-linear rank metric codes and show that it is list decodable as good as random rank metric codes in~\cite{venkat2017}.
%the list decodability of random $\F_q$-linear rank metric codes is shown to match that of a general random rank metric codes.

%Recently, there has also been a lot of interest in the list decoding of rank metric codes. Since there are two representations of rank metric codes (vector representation and matrix representation), this work concerns the list-decodability of $\F_q$ and $\F_{q^m}$-linear self-orthogonal rank metric codes, and establishes trade-offs between list size and the gaps to optimal decoding radius. In this paper, we get two main results with high probability, $\F_q$-linear self-orthogonal rank metric codes are $(\tau, O(\frac{1}{\epsilon}))$-list decodable up to $R=(1-\tau)(1-\rho\tau)-\epsilon.$ For $\F_{q^m}$-linear self-orthogonal rank metric codes can be list decoded up to Gilbert-Varshamov bound with exponential list size.
\subsection*{Motivation}
There has been some interesting findings on the list decodability for random $\F_q$-linear rank metric codes~\cite{venkat2017},~\cite{Gabriele}. An interesting direction is to see whether these new results can be applied to improve results on specific $\F_q$-linear rank metric codes. Due to its potential application in many fields, the specific $\F_q$-linear rank metric codes that we are interested in is the $\F_q$-linear self-orthogonal rank metric codes. Due to the finding that from list decoding point of view, random $\F_q$-linear rank metric codes perform as well as general rank metric codes, a natural question that follows is whether the performance can still be maintained when we further restrict that the random $\F_q$-linear rank metric codes to be also self-orthogonal. Moreover, based on $\F_q$-linear case, we  investigate how well one can list decode random $\F_{q^m}$-linear self-orthogonal rank metric codes.
%~~~~~~G. Nebe et al. investigated self-dual maximum rank distance codes in~\cite{Gabriele}. Inspired by the potential applications of self-orthogonal rank metric codes, it is natural to investigate how well one can list decode random $\F_q$ and $\F_{q^m}$-linear self-orthogonal rank metric codes. Furthermore, we want to determine the optimal trade-off between three parameters: rate, list decoding radius and list size.
%In this chapter, we focus on constructing and list decoding $\F_q$ and $\F_{q^m}$-linear self-orthogonal rank metric codes.

\subsection*{Organization}
Firstly, we give some definitions and notations about self-orthogonal rank metric codes, list decoding and quadratic form. Then, we describe how to construct $\F_q$ and  $\F_{q^m}$-linear self-orthogonal rank metric codes based on the quadratic form and analyze the list decodability of them. Ultimately, we draw a conclusion.
\section{Preliminaries}
\subsection{Defined Self-Orthogonal Rank Metric Codes}
%{\color{blue} Rank metric codes can be considered as a set of $n\times m$ matrices over a finite field $\F_q,$ or equivalently as a set of vectors of length $n$ over the extension field $\F_{q^m}$.
%There exist two representations of rank metric codes; vector representation and matrix representation. In a matrix representation, linear rank metric codes are $\F_q$-linear subspaces of $\F_q^{n\times m}$, where the norm of an element $X\in\F_q^{n\times m}$ is defined as the rank of the matrix. In a vector representation, rank metric codes are $\F_{q^m}$-linear subspaces of the vector space $\F_{q^m}^n,$ where the norm of a vector $x\in \F_{q^m}^n$ is defined as the maximal number of coordinates of $x$ which are linearly independent over $\F_q.$ In the following, we will define $\F_q$-linear and $\F_{q^m}$-linear self-orthogonal rank metric codes, respectively.}
%~~~~~~{\color{blue} There are mainly two different representations of rank metric codes (matrix representation and vector representation). For matrix representation, in which the codewords are assumed to be matrices in $\F_q^{n\times m},$ and vector representation, in which the codewords are assumed to be vectors in $\F_{q^m}^n.$}
Rank metric codes can mainly be interpreted in two different representations. The first representation is to deem each codeword as matrices in $\F_q^{n\times m}.$ Alternatively, we can interpret each element of a rank metric code as a vector in $\F_{q^m}^n.$ 
%{\color{blue} When we consider a rank metric code as a set of matrices, its linearity is considered to be linearity over $\F_q.$} 
In the first representation of codewords as matrices, linear codes are considered to be linear over $\F_q.$ On the other hand, the linearity considered when seeing a rank metric code as a set of vectors is assumed to be $\F_{q^m}$ linearity. The two different representations of rank metric codes provide us with two different ways in defining self-orthogonal rank metric codes ($\F_q$-linear and $\F_{q^m}$-linear).

\subsubsection*{$\F_q$-linear self-orthogonal rank metric codes}
%{\color{blue} Let $q$ be a prime power and denote by $\mathbb{F}_q$ the finite field of $q$ elements. Without loss of generality, we assume that $n\leq m$. A rank metric code is a set $\mC\subseteq \F_q^{n\times m}$ over a field $\F_q.$ For any $X,Y\in \F_{q}^{n\times m}$, the rank distance between $X$ and $Y$ is defined by
%$d_R(X,Y) := \rank(X-Y).$
%The relative rank distance and rate for a rank-metric code $\mC$ by
%\begin{equation*}
%\delta(\mC)=\frac{\min_{X\neq Y\in \mC}\{d_R(X,Y)\}-1}{n}\quad \mathrm{and} \quad R(\mC)=\frac{\log_q |\mC|}{mn}.
%\end{equation*}
%The rate and minimum rank distance of $\mC$ is defined by $R=\frac{log_q|\mC|}{mn}$ and $d_R(\mC)=min\{rank(X-Y)|X\neq Y, X,Y\in\mC\},$ respectively.
%Then, we define the dual code. The dual code of $\mC$ is
%\begin{equation*}
%\mC^\bot=\{X\in\F_q^{n\times m}| Tr(C X^T)=0, \forall C\in\mC\}
%\end{equation*}
%where $T$ is the transpose and $Tr(X)$ is the trace of the matrix $X,$ by using $Tr$ to denote. An $\F_q$-linear rank metric code $\mC$ is self-orthogonal if $\mC\subseteq \mC^\bot$ with dimension $dim_{\F_{q}}(\mathcal{C})\leq \frac{nm}{2}.$ Hence, an  $\F_q$-linear self-orthogonal rank metric code has information rate $0\leq R\leq 1/2.$}

To properly define an $\F_q$-linear self-orthogonal rank metric code, first we briefly provide the definitions and notations of matrix representation for rank metric codes. A rank metric code $\mC$ contains $n\times m$ matrices over $\F_q$ for integers $n,m$ and prime power $q.$ Through the use of matrix transpose, for simplicity, we can just assume that $n$ is at most $m.$ The rank distance between two matrices $X,Y\in\F_q^{n\times m}$ is then defined by the rank norm of the difference between the two, $d_R(X,Y):=\rank(X-Y).$ These parameters can then be used to formulate two relative parameters, namely relative minimum rank distance $\delta$ and rate $R.$ These relative parameters are then defined by
%{\color{blue} For any prime power $q$ and two integers $n$ and $m$  ($n\leq m$), a rank metric code $\mC$ is a set of matrices in $\F_q^{n\times m}.$ The rank distance between two matrices $X,Y\in\F_q^{n\times m}$ is then defined by the rank norm of the distance between the two, $d_R(X,Y):=\rank(X-Y).$ Given these parameters, we can then define two relative parameters, relative minimum rank distance $\delta$ and rate $R$, respectively. Define by,} {\color{red} A rank metric code $\mC$ contains $n\times m$ matrices over $\F_q$ for integers $n,m$ and prime power $q.$ Through the use of matrix transpose, for simplicity, we can just assume that $n$ is at most $m.$ As discussed before, this set is then embedded with a norm and distance function that is based on the rank; $d_R(X,Y):=\rank(X-Y)$ for any $n\times m$ matrices $X$ and $Y$ over $\F_q.$ These parameters can then be used to formulate two relative parameters, namely relative minimum rank distance $\delta$ and rate $R.$ These relative parameters are then defined by}
\begin{equation*}
\delta(\mC)=\frac{\min_{X\neq Y\in \mC}\{d_R(X,Y)\}-1}{n}\quad \mathrm{and} \quad R(\mC)=\frac{\log_q |\mC|}{mn}.
\end{equation*}
The (Delsarte) dual of $\mC$ is then defined to be 
$$\mC^\perp = \{X\in\F_q^{n\times m}| Tr(CX^T) = 0, \forall C\in\mC\}.$$

Based on this definition of dual, an $\F_q$-linear rank metric code $\mC$ is said to be self-orthogonal if $\mC\subseteq \mC^\perp.$ A property of $\F_q$-linear self-orthogonal rank metric code that is readily verified is that its $\F_q$-dimension should be at most $\frac{nm}{2}$ and hence its rate $R$ must be in the range $0\leq R\leq 1/2.$

\subsubsection*{$\F_{q^m}$-linear self-orthogonal rank metric codes}

Pick $\F_{q^m},$ the extension field of $\F_q$ with degree $m.$ The ring isomorphism between $\F_{q^m}$ and $\F_q^m$ through the use of a fixed $\F_q$-basis of $\F_{q^m}$ implies the possibility to identify a rank metric codes over $\F_q^{n\times m}$ as a collection of vectors over $\F_{q^m}^n.$ For any two vectors $x=(x_1,\cdots x_n)$ and $y=(y_1,\cdots,y_n)\in\F_{q^m}^n,$ we say that they are orthogonal to each other if $\langle x,y\rangle = \sum_{i=1}^n x_iy_i=0.$ 
%{\color{blue} If $\langle x,x\rangle=0,$ we say that $x$ is a self-orthogonal vector.}  
The vector $x$ is then called a self-orthogonal vector if it is orthogonal to itself.
The definition of self-orthogonality can  be naturally extended to a set $\{v_1,\cdots, v_t\}$ where this set is self-orthogonal if $\langle v_i,v_j\rangle=0$ for any choices of $i$ and $j,$ $1\leq i, j\leq t.$ 

Clearly, this suggests an alternative way to define a dual $\mC^\perp$ of a rank metric code $\mC\in\F_{q^m}^n,$ namely the collection of codewords that are orthogonal to all codewords in $\mC.$ Analogous to the previous definition, we call an $\F_{q^m}$-linear rank metric code $\mC$ to be self-orthogonal given that $\mC\subseteq \mC^\perp.$ Consequentially, an $\F_{q^m}$-linear rank metric code $\mC$ can only be self-orthogonal if its $\F_{q^m}$-dimension does not exceed $\frac{n}{2}$ or it has rate $R\in(0,1/2).$ We will devote the remainder of this section to investigate when both definition of duals can coincide.

In order to do that, we first need to review some basic concepts of algebra. Recall that for the field $\F_{q^m},$ there are $n$ different $\F_q$-linear automorphisms which are the Frobenius automorphisms $x\mapsto x^{q^{i}}$ for $i=0,\cdots,n-1.$ Based on these automorphisms, we can then define the field trace $tr_{\F_{q^m}/\F_q}$ which sends $x$ to $\sum_{i=0}^{m-1} x^{q^i}.$ We also recall that for any given $\F_q$ basis $\mB=(\beta_1,\cdots, \beta_m)$ of $\F_{q^m},$ there always exists a dual bases $\mB^\ast : = (\beta_1^{\ast} , \cdots \beta_m^{\ast})$ satisfying $tr_{\F_{q^m}/\F_q}(\beta_i\beta_j^\ast)=\delta_{i,j}$ the Kronecker delta function. Now if $\beta_i$ coincides with $\beta_i^\ast$ for all $i,$ $\mB$ is called a self-dual basis. Note that the existence of a self-dual $\F_q$-basis of $\F_{q^m}$ is equivalent to the condition that $q$ is even or both $q$ and $m$ are odd~\cite{Gabriele}. We are ready to conduct our investigations.
\begin{lemma}
%{\color{blue} Take $q$ to be even or both $q$ and $n$ to be odd. Let $\mathcal{B}=(\beta_1,\beta_2,\cdots, \beta_m)$ be a self-dual basis of $\F_{q^m}$ over $\F_q$ and $\mC_1,\mC_2$ be rank metric codes. Then, $Tr(\mC_1\mC_2^T)=0$ with $\mC_1,\mC_2\subseteq\F_{q}^{n\times m}$ if and only if $<\mC_1, \mC_2>=0$ with $\mC_1, \mC_2\subseteq \F_{q^m}^n.$ }
%{\color{blue} Let $q$ be a prime power and $m$ be a positive integer such that either $q$ is even or both $q$ and $m$ are odd. Let $\mB=(\beta_1,\cdots, \beta_m)$ be a self-dual basis of $\F_{q^m}$ over $\F_q$ and consider two linear rank metric codes $\mC_1$ and $\mC_2.$}
Let $q$ and $m$ be chosen such that $\F_{q^m}$ has a self dual basis $\mB=(\beta_1,\cdots, \beta_m).$ Furthermore, pick any two rank metric codes $\mC_1$ and $\mC_2$ over $\F_q^{n\times m}.$ Then 
$$Tr(\mC_1\mC_2^T)=\{Tr(XY^T): X\in \mC_1, Y\in \mC_2\}=\{0\}$$
if and only if 
$$\langle\mC_1,\mC_2\rangle =\{\langle x,y\rangle: x\in \mC_1, y\in\mC_2\}=\{0\}.$$

Note that we assume $\mC_i$ to be in its matrix representation for the earlier equation and vector representation for the latter.

%{\color{red} Then, $\mC_2$ is the trace dual of $\mC_1$ if and only if it is the standard dual of $\mC_1.$ That is, $\mC_2=\{X\in\mathbb{F}_{q}^{n\times m} : Tr(XY^T)=0~ \forall Y\in \mC_1\}$ if and only if $\mC_2=\{x\in\mathbb{F}_{q^m}^n: <x,y>=0 ~\forall y\in\mC_1\}.$}
\end{lemma}
\begin{proof}
Let $\mathbf{a}=(a_1,\cdots,a_n)\in\mC_1$ and $\mathbf{b}=(b_1,\cdots b_n)\in\mC_2.$ We further let $a_i=\sum_{j=1}^m a_{i,j}\beta_j$ and $b_i=\sum_{k=1}^mb_{i,k} \beta_k$ for all $i=1,2\cdots, n.$ Assuming $\langle\mC_1,\mC_2\rangle=\{0\},$ we have 
\begin{eqnarray*} 
0=\langle\mathbf{a}, \mathbf{b}\rangle=\sum_{i=1}^na_{i}b_{i}&=&\sum_{i=1}^n\left(\sum_{j=1}^ma_{i,j}\beta_{j}\right)\left(\sum_{k=1}^mb_{i,k}\beta_{k}\right)=\sum_{i=1}^n\sum_{j=1}^m\sum_{k=1}^ma_{i,j}b_{i,k}\beta_{j}\beta_{k}.
\end{eqnarray*}

Applying the field trace function to both sides, we have 
\begin{align*}
0=tr\langle\mathbf{a}, \mathbf{b}\rangle=\sum_{i=1}^n\sum_{j=1}^m\sum_{k=1}^ma_{i,j}b_{i,k}tr(\beta_{j}\beta_{k})=\sum_{i=1}^n\sum_{j=1}^m\sum_{k=1}^ma_{i,j}b_{i,k}=\sum_{i=1}^n\sum_{j=1}^ma_{i,j}b_{i,j}.
\end{align*}

Note that by assumption, the matrices $A=(a_{i,j}), B=(b_{i,j})$ with size $n\times m$ are the matrix representations of $\mathbf{a}$ and $\mathbf{b}$ respectively. Noting that $Tr(AB^T)=\sum_{i=1}^n\sum_{j=1}^m a_{i,j}b_{i,j},$ the last equality implies that $Tr(AB^T)=0.$ Since $\mathbf{a}$ and $\mathbf{b}$ are arbitrary elements of $\mC_1$ and $\mC_2$ respectively, we have the conclusion $Tr(\mC_1\mC_2^T)=\{0\}$ from the assumption that $\langle\mC_1,\mC_2\rangle=\{0\}.$ 

%Now note that the proof so far tells us that if $B\in\mC_1^\perp,$ then $B\in{\mC_1}^\perp_{Tr}.$ This implies that $\mC_1^\perp\subseteq {\mC_1}^\perp_{Tr}.$ Furthermore the application of Rank-Nullity theorem in both representations yields that $|\mC_1^\perp| = |{\mC_1}^\perp_{Tr}|.$ Consequentially, we must have that $\mC_1^\perp={\mC_1}^\perp_{Tr}.$ 
On the other direction can be easily proved. Suppose that $Tr(AB^T)=0$ for all $A=(a_{i,j})\in\mC_1$ and $B=(b_{i,j})\in\mC_2.$
%{\color{blue}This implies that $B\in \mC_1^\perp.$ {\color{green} This last step does not make sense. We get $B\in\mC_1^\perp$ because $Tr(AB^T)=0.$ Without the previous paragraph plus the fact that now we have trace dual and inner product dual to be the same, how can we go to the next step to say that this implies $\langle A,B\rangle=0?$ (This is why we need both notations in this proof since we have not proved they are the same. In fact, this is the lemma where we prove it is the same and hence after this point on, we always assume trace dual and inner product dual to be one and the same. This is why I believe the result presented here can only apply for the case when either $q$ is even or both $q$ and $m$ are even.)} 
Hence, we have $B\in\mC_1^\perp$ or $\langle A,B\rangle=0$ for all $A\in\mC_1$ and $B\in\mC_2$ which ultimately implies that $\langle\mC_1,\mC_2\rangle=\{0\}.$
\end{proof}

%The rank of a vector $v=(v_1, v_2,\cdots, v_n)\in\mathbb{F}_{q^m}^n$ is defined as the maximum number of coordinates in $v$ that are linearly independent over $\mathbb{F}_q,$ i.e. $rank(v):=dim_{\mathbb{F}_q}<v_1, v_2, \cdots, v_n>.$ Then, we have a rank metric distance given by $d(v,u)=rank(v-u)$ for $v,u\in \mathbb{F}_{q^m}^n.$ 
%{\color{blue}  The dual code $\mathcal{C^\bot}$ of an $\F_{q^m}$-linear rank metric code $\mathcal{C}$ consists of all vectors in $\mathbb{F}_{q^m}^n$ that are orthogonal to every codeword in $\mathcal{C}.$ A subset $\{v_1, v_2,\cdots, v_t\}$ of $\mathbb{F}_{q^m}^n$ is called self-orthogonal if $<v_i,v_j>=0$ for all $1\leq i, j\leq t.$ An $\F_{q^m}$-linear rank metric code $\mathcal{C}$ is said to be self-orthogonal if $\mathcal{C}\subseteq \mathcal{C^\bot}$ with dimension $dim_{\F_{q^m}}(\mathcal{C})\leq \frac{n}{2}.$ Hence, an $\F_{q^m}$-linear self-orthogonal rank metric code has information rate $0\leq R\leq 1/2.$}

%%%%%%%%

%===============================
\subsection{List Decoding}
Given a received word, list decoding outputs a list of codewords, if successful, it contains the correct transmitted codeword.
Analogous to the Hamming ball in classical block codes, in rank metric codes, we also have the concept of rank metric ball.
The formal definition is given in the following.
\begin{defn} Let $\tau\in(0,1)$ and $X\in \F_{q}^{n\times m}$. The rank metric ball centre at $X$ and radius $\tau n$ is defined by 
%For the real $\tau\in (0,1)$, matrix $X\in \F_{q}^{n\times m}$ as a center of the rank metric ball, and radius $\tau n$ is defined by
%the rank metric ball with center $X\in \F_{q}^{n\times m}$ and radius $\tau n$ is defined by
\begin{equation*}
\mB_{R}(X,\tau n):=\{Y\in \F_{q^m}^n:d_R(X,Y)\leq \tau n\}.
\end{equation*}
\end{defn}
%Next we recall the Gaussian binomial coefficient to help us in estimating the size of a rank metric ball.

For any $n$-dimensional vector space $V$ over $\F_q,$ we denote by ${n \brack k}_q,$ the number of subspaces of $V$ with dimension $k.$ This is called the Gaussian binomial coefficient and it has the following explicit formula,
\[
{n\brack k}_{q}=\prod_{i=0}^{k-1}\frac{q^{n}-q^{i}}{q^{k}-q^{i}}.
\]
It can be verified that this formula ${n\brack k}_q$ has the following bounds that can be used as estimation~\cite{Gad2008},
\begin{equation*}
q^{k(n-k)}\leq{n\brack k}_{q}\leq 4 q^{k(n-k)}.
\end{equation*}
%Now we move to the list decodability of a rank metric code.
\begin{defn} For an integer $L\geq 1$ and a real $\tau\in (0,1)$,
a rank metric code $\mC$ is said to be $(\tau n, L)$-list decodable if for every $X\in \F_q^{n\times m}$
\begin{equation*}
|\mB_R(X,\tau n)\cap \mC|\leq L.
\end{equation*}
Analogously, for  $x\in\F_{q^m}^n,$ 
\begin{equation*}
|\mB_R(x,\tau n)\cap \mC|\leq L.
\end{equation*}
\end{defn}

\subsection{Quadratic forms}
%{\color{blue} An $n$-variate quadratic form over $\mathbb{F}_q$ is a homogeneous polynomial of degree $2$ in $n$ variables with coefficients in $\mathbb{F}_q,$ i.e.}
We say $f(x)$ is an $n$-variate quadratic form over $\F_q,$ if it is a degree $2$ homogeneous multinomial of $n$ variables with coefficients from $\F_q.$ The general formula that $f$ should follow 
\begin{equation*}
f(x)=f(x_1,x_2,\cdots,x_n)=\sum_{i,j=1}^na_{ij}x_ix_j,a_{ij}\in\mathbb{F}_q.
\end{equation*}

%{\color{blue} This form can be rewritten by using matrices. Let $x$ be the column vector with components $x_1,\cdots, x_n$ and $n\times n$ matrix $A=[a_{ij}]$ over $\F_q$, where every entries $a_{ij}\in\mathbb{F}_q.$ Then, we have $f(x)=x^TAx.$}
Note that an $n$-variate quadratic form $f$ over $\F_q$ can be expressed as multiplication of matrices. Assuming $x=(x_1,\cdots, x_n)^T$ and $A=(a_{i,j})_{i,j=1,\cdots, n}$ over $\F_q,$ $f(x)$ can then be rewritten as $f(x) = x^T Ax.$ 
%{\color{blue} If there exists an invertible matrix $M$ with size $n\times n$ over $\F_q$ such that the quadratic form $f(xM)$ is equal to $g(x)$, then we can say that the two quadratic forms $f(x)$ and $g(x)$ are equivalent. Two equivalent quadratic forms contain the same number of zeros. } 

Two quadratic forms $f$ and $g$ of $n_1$ and $n_2$ indeterminates respectively are called equivalent provided that we can find a full rank $n_1\times n_2$ matrix $M$ over $\F_q$ satisfying $f(xM) = g(x).$ Note that equivalence implies the same number of roots.

%{\color{blue} For a nonzero quadratic form $f(x),$ the smallest number $m$ for which $f(x)$ is not equivalent to a quadratic form in fewer than $m$ indeterminates is called the rank of $f(x).$ The rank of the zero quadratic form is defined to be $0.$ When the rank of a nonzero quadratic form $f(x)$ is equal to $n,$ then we call $f(x)$ non-degenerate.}
Equivalence enables two quadratic forms of different number of indeterminates to be closely related to each other. Given a non-zero quadratic form $f(x),$ the smallest number of indeterminates that a quadratic form $g(x)$ can have while still being equivalent to $f(x)$ is a parameter of $f(x)$ that is called to be its rank. By convention, we let the zero quadratic form be rank $0.$ A non-zero quadratic form $f(x)$ is said to be non-degenerate if its rank is equal to the number of its indeterminates.
%A fundamental problem in the theory of quadratic form is how much one can simplify $f(x)$ by means of nonsingular linear transformation of indeterminates. 
%In the theory of quadratic form, we want to know how many one can simplify $f(x)$ by means of nonsingular linear transformation of indeterminates. For any finite field $\F_q,$ two quadratic forms $f(x)$ and $g(x)$ over finite field $\F_q$ are called equivalent if $f(x)$ can be transformed into $g(x)$ by means of a nonsingular linear substitution of indeterminates~\cite{finitefields}. 

%In other words, two equivalent quadratic forms have the same number of zeros. For a nonzero quadratic form $f(x),$ the smallest number $m$ for which $f(x)$ is not equivalent to a quadratic form in fewer than $m$ indeterminates is called the rank of $f(x)$ in~\cite{L.F.J2015}. If the rank of a nonzero quadratic form $f(x)$ is $n,$ then we called $f(x)$ as non-degenerate.

%{\color{blue}In this paper, we are interested in the number of solutions of $f(x)=0$ for a quadratic form $f(x)$. 
%Since we discuss self-orthogonal rank metric codes with two representations, {\color{blue} so the following lemma we combine %several results in~\cite{finitefields} and  write down as two cases (over $\F_q$ and $\F_{q^m}$).} 
%We combine several results in~\cite{finitefields} as a lemma that considers two cases (over $\F_q$ and $\F_{q^m}$).}
To aid our analysis in this paper, the number of roots a quadratic form $f(x)$ is a topic of interest. We combine several results in~\cite{finitefields} as a lemma to consider two cases (over $\F_q$ and $\F_{q^m}$).
\begin{lemma}~~\cite{finitefields}~\label{lemma:4.1}
 Let $f(x):=f(x_1,x_2,\cdots,x_n)$ be a quadratic form with rank $r$ over $\mathbb{F}_{q}$ (or $\F_{q^m}$). $N(f(x)=0)$ denotes  the number of roots of $f(x)=0$ in $\mathbb{F}_{q}^{n\times m}$ (or $\F_{q^m}^n$). If $r=0,$ then we have $N(f(x)=0)=q^{mn}.$ If $r\geq 1,$ then we have the following results:
%{\color{red} Consider a rank $r$ quadratic form over $\F_q$ (or $\F_{q^m}$), $f(x)$ where $x=(x_1,\cdots, x_n).$ Let $N(f(x)=0)$ be a non-negative integer indicating the number of possible $x\in\F_q^{n\times m}$ (or $\F_{q^m}^n$). If $r=0,$ then $N(f(x)=0)=q^{mn}$ Otherwise, we have the following result. }
%\begin{eqnarray*}
%N(f(x)=0)=
%\begin{cases}
%q^{m(n-1)}, &r~ is~ odd.\cr q^{m(n-1)}\pm(q^m-1)q^{m(n-r/2-1)}, &r~ is~ even.
%\end{cases}
%\end{eqnarray*}}
%{\color{red} Let $f(x)$ be a quadratic form defined over $\F_q$ (or $\F_{q^m}$) with rank $r.$ Denote by $N(f(x)=0)$ the number of solutions of $f(x)=0$ where $x\in\F_q^{n\times m}$ (or $\F_{q^m}^n$). If $r=0,$ then $N(f(x)=0)= q^{mn}$ in both cases. If $r\geq 1,$ ($r\leq mn$ if $f$ is defined over $\F_q$ and $r\leq n$ if $f$ is defined over $\F_{q^m}$) then we have the following results. }

If $f(x)$ is defined over $\F_q$ with solutions in $\F_q^{n\times m},$ 
\begin{eqnarray*}
N(f(x)=0)=
\begin{cases}
q^{mn-1}, &r~ is~ odd.\cr q^{mn-1}\pm(q-1)q^{mn-r/2-1}, &r~ is~ even.
\end{cases}
\end{eqnarray*}
Alternatively, if $f(x)$ is defined over $\F_{q^m}$ with solutions in $\F_{q^m}^n,$ 
\begin{eqnarray*}
N(f(x)=0)=
\begin{cases}
q^{m(n-1)}, &r~ is~ odd.\cr q^{m(n-1)}\pm(q^m-1)q^{m(n-r/2-1)}, &r~ is~ even.
\end{cases}
\end{eqnarray*}

\end{lemma}

%%%%%%%%%%%%%%%%%%%%%%%%%%%%%%%%%%%%%%

%%%%%%%%%%%%%%%%%%
\section{Construction of Random Self-Orthogonal Rank Metric Codes}
\subsection{Construct $\F_q$-linear Self-Orthogonal Rank Metric Codes}
In this part, we construct $\F_q$-linear self-orthogonal rank metric codes based on quadratic forms.

Let $A=[a_{ij}], 1\leq i\leq n, 1\leq j\leq m$ be a word. If $A$ is self-orthogonal, then 
\begin{equation*}
Tr(AA^T)=\sum_{i=1}^{n}\sum_{j=1}^{m}a_{ij}a_{ij}=0.
\end{equation*}
%Consider this bijection
%$\psi:\mathbb{Z}_m\times\mathbb{Z}_n\to \mathbb{Z}_{mn}.$
%We can rewrite
Considering the standard bijection from $[n]\times [m]$ to $[nm],$ where $[n]=\{0,\cdots, n-1\}, [m]=\{0,\cdots, m-1\},$ we can rewrite the double index $(i,j)$ to a single index to obtain
 \begin{eqnarray*}
 Tr(AA^T)&=&\sum_{i=1}^{n}\sum_{j=1}^{m}a_{ij}a_{ij}=\sum_{\ell=1}^{mn}a_\ell^2=0.
 \end{eqnarray*}
 {\bf Construction}
 \begin{itemize}
 \item{Step $1.$ Choose a nonzero random solution $A_1\in\mathbb{F}_{q}^{n\times m}$ of the quadratic equation $x_1^2+x_2^2+\cdots+x_{mn}^2=0.$} 
 
By Lemma~\ref{lemma:4.1} we can obtain that the above equation has at least $q^{mn-2}$ solutions, so a self-orthogonal word $A_1$ can be found.

%Firstly, we choose a nonzero random solution $A_1\in\mathbb{F}_{q}^{n\times m}$ of the quadratic equation $x_1^2+x_2^2+\cdots+x_{mn}^2=0.$ Since this equation has at least $q^{mn(n-2)}$ solutions by Lemma~\ref{lemma:4.1}, so a self-orthogonal $A_1$ can be found. 

%{\color{blue} Confirm we can obtain a linearly independent set $\{A_1, A_2, \cdots, A_{k-1}, A_k\}$ of random self-orthogonal matrices.} {\color{red} \textbf{Induction step} 
\item {Step $2.$ Obtain a linearly independent set $\{A_1,A_2,\cdots,A_{k-1},A_k\}$ of random self-orthogonal matrices given $A_1,\cdots,A_{k-1}.$}

{Firstly, we assume that a linearly independent set $\{A_1, A_2,\cdots, A_{k-1}\}$ of random self-orthogonal matrices has already been found, i.e. $Tr(A_iA_j^T)=\sum_{\ell=1}^{mn}a_{i_\ell}a_{j_\ell}=0, 1\leq i,j \leq k-1$. Then, if we want to find the $k$-th matrix $A_k,$ then we need to find a solution of the following equations
\begin{equation}
\left\{
\begin{aligned}~\label{eq:4.1}
a_{11}x_1+a_{12}x_2+\cdots+a_{1,mn}x_{mn}&=&0, \\
\vdots \\
a_{{k-1},1}x_1+a_{{k-1},2}x_2+\cdots+a_{{k-1},mn}x_{mn}&=&0,\\
x_1^2+x_2^2+\cdots+x_{mn}^2&=&0.
\end{aligned}
\right.
\end{equation}}
 \end{itemize}
Take the first $k-1$ equations above in to the last one, we have a quadratic equation $g(x_{i_1}, x_{i_2},\cdots, x_{i_{mn-k+1}})$ of $mn-k+1$ variables. So, $N(g(x_{i_1}, x_{i_2},\cdots, x_{i_{mn-k+1}})=0)$ is the number of solutions of the equation~(\ref{eq:4.1}). The number of  the cardinality of $\rm{span}\{A_1, A_2, \cdots, A_{k-1}\}$ is equal to $q^{mn(k-1)}.$ And, $N(g(x_{i_1}, x_{i_2},\cdots, x_{i_{mn-k+1}})=0)>q^{mn(k-1)},$ thus we can randomly choose a solution $A_k$ from~(\ref{eq:4.1}), it is not contained in $\rm{Span}\{A_1, A_2, \cdots, A_{k-1}\}.$ 

So, we can obtain a linearly independent set $\{A_1, A_2, \cdots, A_{k-1}, A_k\}$ of random self-orthogonal matrices. 

Moreover, by Lemma~\ref{lemma:4.1}, the number of solution $N(g(x_{i_1}, x_{i_2},\cdots, x_{i_{mn-k+1}})=0)$ of $g(x_{i_1}, x_{i_2},\cdots, x_{i_{mn-k+1}})=0$ is at least $q^{mn-k-1}.$ Thus, {the set can always be constructed as long as} $k\leq(mn-1)/2.$
%\subsection{Construction of Random $\F_{q^m}$-linear Self-Orthogonal Rank Metric Codes}
\subsection{Construct $\F_{q^m}$-linear Self-Orthogonal Rank Metric Codes}
We study how to construct $\F_{q^m}$-linear self-orthogonal rank metric codes. The idea is similar to the construction of $\F_{q}$-linear self-orthogonal rank metric codes. Constructing a random $\F_{q^m}$-linear self-orthogonal rank metric code is equivalent to find a linearly independent set $\{x_1, x_2,\cdots, x_k\}$ of random $\F_{q^m}$-linear self-orthogonal vectors, where $x_i\in\F_{q^m}, 1\leq i\leq k.$ 

Choose a nonzero random solution $x_1=(x_{11}, x_{12},\cdots, x_{1n})\in\mathbb{F}_{q^m}^n$ of the quadratic equation $z_1^2+z_2^2+\cdots+z_n^2=0.$ This equation has at least $q^{m(n-2)}$ roots, so a self-orthogonal $x_1$ can be found. The same method as the construction $\F_q$-linear self-orthogonal rank metric codes can be conducted. Then, we can confirm there exists a linearly independent set $\{x_1, x_2, \cdots, x_{k-1}, x_k\}$ of $\F_{q^m}$-linear self-orthogonal vectors.  In addition, by calculation we have such $x_k$ as long as $k\leq(n-1)/2.$

\section{List Decoding Self-Orthogonal Rank Metric Codes}

\subsection{List Decoding $\F_{q}$-linear Self-Orthogonal Rank Metric Codes}
In this part, we investigate the list decodability of $\F_{q}$-linear self-orthogonal rank metric codes. We show that rate and decoding radius of $\F_{q}$-linear self-orthogonal rank metric codes can achieve the Gilbert-Varshamov bound. From now on, the information rate $\frac{log_q|\mC|}{mn}$ and the ratio $\frac{n}{m}$ are denoted by $R$ and $\rho,$ respectively.

Our main result of list decoding of $\F_q$-linear self-orthogonal rank metric codes can be found in the Theorem~\ref{thm:4.6}. With the help of studying and discussing the weight distribution of certain rank metric code, we can deal with it.

%In order to prove this theorem, we firstly consider reducing the problem of the list decodability of an $\F_q$-linear self-orthogonal rank metric code to studying the weight distribution of certain rank metric code containing a given $\F_q$-linear self-orthogonal rank metric code. 

%We cite a result from~\cite{venkat2017}, where $B_R(X,\tau n)$ presents the rank metric ball with center $X\in\F_q^{n\times m}$ and certain radius $\tau n.$
\begin{lemma}~\cite{venkat2017}\label{lemma:4.7}
For all integers $n\leq m$, every $\tau\in(0,1)$ and $\ell=O(\sqrt{nm}),$ there exists a constant $C_{\tau,q}>1$ such that if $X_1,\cdots, X_\ell$ are selected independently and uniformly at random from $B_R(0,\tau n),$ then we have
\begin{equation*}
Pr[\mid \rm{span}\{X_1,\cdots, X_\ell\}\cup B_R(0,\tau n)\mid\geq C_{\tau,q}\cdot \ell]\leq q^{-(3-O(1))mn}
\end{equation*}
\end{lemma}
From the above lemma, it reveals that randomly picking $\ell$ words from $B_R(0, \tau n),$ there exists more than $\Omega(\ell)$ words in the span of $\ell$ words lies in the $B_R(0, \tau n)$ happens with a very small probability, where the parameter $\ell$ depends on the list size $L.$
%The above lemma shows that if we randomly pick $\ell$ words from the rank metric ball $B_R(0,\tau n),$ where $\ell$ depends on the list size $L$, the probability that more than $\Omega(\ell)$ words in the span of these $\ell$ words lies in the ball $B_R(0, \tau n)$ is quite small.

Then, we consider the following result on the probability that a random $k$ dimension $\F_q$-linear rank metric code contains a $k-1$ dimension $\F_q$-linear self-orthogonal subcode and a given set $\{X_1,\cdots,X_\ell\}\subseteq \F_q^{n\times m}$ of linearly independent vectors. Let $\mC^*_k$ present the set of $k$ dimension $\F_q$-linear rank metric codes, where every code contains a $k-1$ dimension $\F_q$-linear self-orthogonal subcode.

\begin{lemma}~\cite{L.F.J2015}\label{lemma:4.8}
For any $\mathbb{F}_{q}$-linearly independent words $X_1, X_2,\cdots, X_\ell$ in $\F_{q}^{n\times m}$ with $\ell\leq k <mn/2,$ the probability of a random code $\mathcal{C}^*$ from $\mathcal{C}^*_k$ contains $\{X_1,X_2,\cdots, X_\ell\}$ is
\begin{equation}
Pr_{\mathcal{C}^*\in \mathcal{C}^*_k}[\{X_1,X_2,\cdots, X_\ell\}\subseteq \mathcal{C}]\leq\left\{
\begin{aligned}
{q}^{((k+\ell-mn-1)\ell+2k-1)}, ~~if~~ q~~ is~~ even;\\
{q}^{((k+\ell-mn-2)\ell+4k-2)}, ~~if~~ q~~ is~~ odd.
\end{aligned}
\right.
\end{equation}
Thus, we have 
\begin{equation}
Pr_{\mathcal{C}^*\in \mathcal{C}^*_k}[\{X_1,X_2,\cdots, X_\ell\}\subseteq \mathcal{C}]\leq {q}^{((k+\ell-mn-2)\ell+4k-1)}.
\end{equation}
\end{lemma}

Based on the Lemma~\ref{lemma:4.7} and Lemma~\ref{lemma:4.8}, we prove Theorem~\ref{thm:4.6}.
\begin{theorem}~\label{thm:4.6}
Let $q$ be prime power and $\tau\in(0,1).$
There exist a constant $M$ and all large enough $n$, for small $\epsilon>0,$  an $\F_q$-linear self-orthogonal rank metric code $\mC\subseteq \F_q^{n\times m}$ of rate $R=(1-\tau)(1-\rho\tau)-\epsilon$ is ${(\tau n, O_{\tau, q}(\frac{1}{\epsilon}))}$-list decodable with high probability at least $1-q^{-2mn}$.
\end{theorem}

%Firstly, in order to prove the Theorem~\ref{thm:4.6}, we need to reduce the problem of the list decodability of an $\F_q$-linear self-orthogonal rank metric code to study the weight distribution of certain $\F_q$-linear rank metric code containing a given $\F_q$-linear self-orthogonal rank metric code. 

\begin{proof}
Pick $M=5C_{\tau,q},$ where $C_{\tau,q}$ is the constant in Lemma~\ref{lemma:4.7}.
Set $L=\lceil{\frac{M}{\epsilon}}\rceil$ and $n$ to be large enough.

Let $\mC$ be an $\F_{q}$-linear self-orthogonal rank metric codes with $\dim_{\F_{q}}=Rmn$ in $\F_{q}^{n\times m}$, the size $|\mC|=q^{Rmn}.$ We want to show that with high probability, $\mC$ is  ${(\tau, O_{\tau,q}(\frac{1}{\epsilon}))}$-list decodable. In other words, the code $\mC$ is not  ${(\tau n,  O_{\tau,q}(\frac{1}{\epsilon}))}$-list decodable with low probability, i.e.,
\begin{equation}~\label{eq:4.2}
Pr_{\mC\in\mC_{Rmn}}[\exists X\in\F_{q}^{n\times m},~|B_R(X,\tau n)\cap \mC|\geq L]< q^{-2mn},
\end{equation}
where $\mC_{Rmn}$ denotes the set of $\F_{q}$-linear self-orthogonal rank metric codes with dimension $Rmn.$

Let $X\in\F_{q}^{n\times m}$ be picked uniformly at random, define
\begin{equation*}
\triangle:=Pr_{{\mC\in\mC_{Rmn}}, X\in\F_{q^{n\times m}}}[|B_R(X,\tau n)\cap\mC|\geq L]
\end{equation*}
For~(\ref{eq:4.2}), it suffices to prove
\begin{equation}~\label{eq:4.3}
\triangle< q^{-2mn}\cdot q^{-(1-R)mn}
\end{equation}
The inequality~(\ref{eq:4.3}) is derived from~(\ref{eq:4.2}). For every $\F_{q}$-linear code $\mC$, due to "bad" case $X$ such that ($|B_R(X,\tau n)\cap\mC|\geq L$), there are $q^{Rmn}$ such ''bad'' $X.$

%since for every $\F_{q}$-linear $\mC$ for which there is a "bad" $X$ such that $|B_R(X,\tau n)\cap\mC|\geq L,$ there are $q^{Rmn}$ such ''bad'' $X.$

Since $\mC$ is $\F_q$-linear, we have 
\begin{align*}
\triangle=&Pr_{{\mC\in\mC_{Rmn}}, {X\in\F_{q}^{n\times m}}}[|B_R(X,\tau n)\cap\mC|\geq L]\\
=&Pr_{{\mC\in\mC_{Rmn}}, {X\in\F_{q}^{n\times m}}}[|B_R(0,\tau n)\cap(\mC+X)|\geq L]\\
\leq&Pr_{{\mC\in\mC_{Rmn}}, {X\in\F_{q}^{n\times m}}}[|B_R(0,\tau n)\cap \rm Span_{\F_{q}}(\mC+X)|\geq L]\\
\leq&Pr_{{\mC^*\in\mC^*_{Rmn+1}}}[|B_R(0,\tau n)\cap \mC^*|\geq L],
\end{align*}
where $\mC^*$ is a $Rmn+1$ dimension random $\F_{q}$-subspace of $\F_{q}^{n\times m}$ containing $\rm Span_{\F_{q}}(\mC+X)$ ( If $X$ is not in $\mC,$ then $\mC^*=\rm Span_{\F_{q}}(\mC, X);$ otherwise $\mC^*=\rm Span_{\F_{q}}(\mC, Y),$ where word $Y$ is  randomly picked from $\F_{q}^{n\times m}\backslash \mC$).

For each integer $\ell \in[log_{q} L, L],$ let $\mathcal{F}_\ell$ be the set of all tuples $(X_1, X_2, \cdots, X_\ell)\in B_R(0,\tau n)^\ell$ such that $X_1, X_2, \cdots, X_\ell$ are linearly independent and $|\rm span(X_1, \cdots, X_\ell)\cap B_R(0,\tau n)|\geq L.$

Let $\mathcal{F}=\bigcup_{\ell=\lceil{loq_{q}L} \rceil}^L \mathcal{F}_\ell.$ For each $X=(X_1,\cdots, X_\ell)\in\mathcal{F},$ let $\{X\}$ and $(X)$ denote the set $\{X_1,\cdots,X_\ell\}$ and the tuple $(X_1,\cdots, X_\ell),$ respectively.

Claim that if $|B_R(0,\tau n)\cap \mC^*|\geq L,$ there must exist $(X)\in\mathcal{F}$ such that $\{X\}\subseteq \mC^*.$ Indeed, let $\{H\}$ be a maximal linearly independent subset of $B_R(0,\tau n)\cap \mC^*.$ If $|\{H\}|<L,$ then we have $\{X\}=\{H\}$. Otherwise, we have $\{X\}$ to be any subset of $\{H\}$ of size $L.$
Thus, 
\begin{align*}
\triangle\leq&Pr_{{\mC^*\in\mC^*_{Rmn+1}}}[|B_R(0,\tau n)\cap \mC^*|\geq L]\\
\leq&\sum_{(X)\in\mathcal{F}_\ell}Pr_{{\mC^*\in\mC^*_{Rmn+1}}}[\{X\}\subseteq\mC^*]\\
=&\sum_{\ell=\lceil{loq_{q}L} \rceil}^L\sum_{(X)\in\mathcal{F}_\ell}Pr_{{\mC^*\in\mC^*_{Rmn+1}}}[\{X\}\subseteq\mC^*]\\
=&\sum_{\ell=\lceil{loq_{q}L} \rceil}^L|\mathcal{F}_\ell|Pr_{{\mC^*\in\mC^*_{Rmn+1}}}[\{X\}\subseteq\mC^*]\\
\leq&\sum_{\ell=\lceil{loq_{q}L} \rceil}^L|\mathcal{F}_\ell|{q}^{(((Rmn+1)+\ell-mn-2)\ell+4(Rmn+1)-1)}.%%
\end{align*}
The last inequality is from the Lemma~\ref{lemma:4.8}. We want to get a good bound of our probability, so we need to take a reasonable good upper bound for $|\mathcal{F}_\ell|.$ In~\cite{venkat2017}, we bound $|\mathcal{F}_\ell|$ relying on the value of the parameter $\ell$.
\begin{itemize}
\item{Case $1.$ $\ell<\frac{5}{\epsilon}$}
\\In this case, we have $\frac{|\mathcal{F}_\ell|}{|B_R(0,\tau n)|^\ell}$ is a lower bound on the probability that matrices $X_1, X_2, \cdots, X_\ell$ chosen independently and uniformly at random from the rank metric ball $B_R(0,\tau n)$ are
\begin{equation*}
\mid \rm span\{X_1,\cdots, X_\ell\}\cup B_R(0,\tau n)\mid\geq L.
\end{equation*}
By Lemma~\ref{lemma:4.7}, the probability is at most $q^{-2mn}$, thus
\begin{equation*}
|\mathcal{F}_\ell|\leq |B_R(0,\tau n)^\ell|\cdot q^{-2mn}\leq\left(4q^{mn(\tau+\tau \rho-\tau^2\rho)}\right)^\ell\cdot q^{-2mn}.
\end{equation*}
\item{Case $2.$ $\ell\geq\frac{5}{\epsilon}$}
\\We have the simple bound of 
$$|\mathcal{F}_\ell|\leq |B_R(0,\tau n)^\ell|\leq\left(4q^{mn(\tau+\tau \rho-\tau^2\rho)}\right)^\ell.$$ 
\end{itemize}

Finally, taking the value of $R=(1-\tau)(1-\rho\tau)-\epsilon$ into the below inequality, 
\begin{align*}
\triangle
\leq&\sum_{\ell=\lceil{loq_{q}L} \rceil}^L|\mathcal{F}_\ell|{q}^{(((Rmn+1)+\ell-mn-2)\ell+4(Rmn+1)-1)}\\
\leq&\sum_{\ell=\lceil{loq_{q}L} \rceil}^{\lceil\frac{5}{\epsilon}\rceil-1}|\mathcal{F}_\ell|{q}^{(((Rmn+1)+\ell-mn-2)\ell+4(Rmn+1)-1)}+\sum_{\ell=\lceil\frac{5}{\epsilon}\rceil}^L|\mathcal{F}_\ell|{q}^{(((Rmn+1)+\ell-mn-2)\ell+4(Rmn+1)-1)}\\%%
%\leq&\sum_{\ell=\lceil{loq_{q}L} \rceil}^{\lceil\frac{5}{\epsilon}\rceil-1} \left(4q^{mn(\tau+\tau \rho-\tau^2\rho)}\right)^\ell\cdot q^{-2mn}{q}^{(((Rmn+1)+\ell-mn-2)\ell+4(Rmn+1)-1)}\\
%+&\sum_{\ell=\lceil\frac{5}{\epsilon}\rceil}^L\left(4q^{mn(\tau+\tau \rho-\tau^2\rho)}\right)^\ell{q}^{(((Rmn+1)+\ell-mn-2)\ell+4(Rmn+1)-1)}\\%%
=&q^{-2mn}q^{4Rmn}\cdot\sum_{\ell=\lceil{loq_{q}L} \rceil}^{\lceil\frac{5}{\epsilon}\rceil-1} 4^\ell q^{{mn\ell}(\tau+\tau \rho-\tau^2\rho+R-1)}+q^{4Rmn}\cdot\sum_{\ell=\lceil\frac{5}{\epsilon}\rceil}^L4^\ell q^{mn\ell(\tau+\tau \rho-\tau^2\rho+R-1)}\\%%
=&q^{-2mn}q^{4Rmn}\cdot\sum_{\ell=\lceil{loq_{q}L} \rceil}^{\lceil\frac{5}{\epsilon}\rceil-1} 4^\ell q^{{mn\ell}(-\epsilon)}+q^{4Rmn}\cdot\sum_{\ell=\lceil\frac{5}{\epsilon}\rceil}^L4^\ell q^{mn\ell(-\epsilon)}\\%%
\leq&q^{-2mn}q^{4Rmn}\cdot\sum_{\ell=\lceil{loq_{q}L} \rceil}^{\lceil\frac{5}{\epsilon}\rceil-1} 4^\ell q^{{mn\ell}(-\epsilon)}+q^{4Rmn}\cdot\sum_{\ell=\lceil\frac{5}{\epsilon}\rceil}^L4^\ell q^{-5mn}\\
\leq&q^{-2mn}.
\end{align*}
Thus, an $\F_q$-linear self-orthogonal rank metric code with rate  $R=(1-\tau)(1-\rho\tau)-\epsilon$ is not ${(\tau n,  O_{\tau,q}(\frac{1}{\epsilon}))}$-list decodable with an exponential small probability $q^{-2mn}.$

\end{proof}
\subsection{List Decoding $\F_{q^m}$-linear Self-Orthogonal Rank Metric Codes}
We consider the probability that a random dimension $k$ $\mathbb{F}_{q^m}$-linear code contains a self-orthogonal $k-1$ dimension subcode and a given set $\{v_1,v_2,\cdots, v_\ell\}\subseteq\mathbb{F}_{q^m}^n$ of linearly independent vectors. Let $\mathcal{C}^*_k$ present the set of ${\F_{q^m}}$-linear codes in which every code contains an $\F_{q^m}$-linear dimension  $k-1$ self-orthogonal subcode.

\begin{lemma}~\cite{L.F.J2015}
For any $\mathbb{F}_{q^m}$-linearly independent vectors $x_1, x_2,\cdots, x_\ell$ in $\mathbb{F}_{q^m}^n$ with $\ell\leq k <n/2,$ the probability of a random code $\mathcal{C}^*$ from $\mathcal{C}^*_k$ contains $\{x_1,x_2,\cdots, x_\ell\}$ is
\begin{equation}
Pr_{\mathcal{C}^*\in \mathcal{C}^*_k}[\{x_1,x_2,\cdots, x_\ell\}\subseteq \mathcal{C}]\leq {q^{m}}^{((k+\ell-n-2)\ell+4k-1)}.
\end{equation}
\end{lemma}
\begin{theorem}
Let $q$ be prime power and a real $\tau\in(0,1).$ There exist a constant $M$ and all large enough $n$, for small $\epsilon>0,$ an $\F_{q^m}$-linear self-orthogonal rank metric code $\mC\subseteq \F_{q^m}^{n}$ of $R=(1-\tau)(1-\rho\tau)-\epsilon$ is ${(\tau n, \exp(O_{\tau, q}(\frac{1}{\epsilon})))}$-list decodable with high probablility $1-q^{-2mn}.$
%For every prime power $q$ and a real $\tau\in(0,1).$ There exists a constant $M,$ such that for small $\epsilon>0$ and all large enough $n,$ an $\F_q$-linear self-orthogonal rank metric code $\mC\subseteq \F_q^{n\times}$ of $R=(1-\tau)(1-\rho\tau)-\epsilon$ is ${(\tau, \frac{M}{\epsilon})}$-list decodable with high probablility $1-q^{-mn}$.
\end{theorem}
%%
%\begin{theorem}
%Take $q$ be even or both $q$ and $n$ are odd. Let $b=\frac{n}{m}\leq 1$. For real  $\tau, R\in(0,1)$ satisfying $0\leq R\leq min\{\frac{1}{2}, (1-\tau)(1-b\tau)\}$, there exists a constant $M_\delta,$ such that for small $\epsilon>0$ and all large enough $n,$ an $\F_{q^m}$-linear self-orthogonal rank metric code $\mC\subseteq\F_{q^m}^n$ of rate  $R=\frac{(1-\tau)(1-b \tau)}{1+\epsilon}-\epsilon$ is $(\tau, O(exp(\frac{1}{\epsilon})))$-list decodable with probability $1-q^{-mn}$.
%\end{theorem}
\begin{proof}
Put $L=\lceil{\frac{1}{\epsilon}}\rceil,$ $n$ is large enough. Let $\mC$ be an $\F_{q^m}$-linear self-orthogonal rank metric codes with $\dim_{\F_{q^m}}=Rn$ in $\F_{q^m}^n$, the size $|\mC|=q^{Rmn}.$ We want to show that $\mC$ is not ${(\tau, \exp( O_{\tau,q}(\frac{1}{\epsilon})))}$-list decodable, i.e.,
\begin{equation}~\label{eq:4.42}
Pr_{\mC\in\mC_{Rn}}[\exists x\in\F_{q^m}^n, ~|B_R(x,\tau n)\cap \mC|\geq L]<q^{-2mn},
\end{equation}
where $\mC_{Rn}$ denotes the set of $\F_{q^m}$-linear self-orthogonal rank metric codes with dimension $Rn.$

Let $x\in\F_{q^m}^n$ be picked uniformly at random, define
\begin{equation*}
\triangle:=Pr_{{\mC\in\mC_{Rn}}, {x\in\F_{q^m}^n}}[|B_R(x,\tau n)\cap\mC|\geq L]
\end{equation*}
To prove inequality~(\ref{eq:4.42}), we need to show that 
\begin{equation}~\label{eq:4.43}
\triangle<q^{-2mn}\cdot q^{-(1-R)mn}
\end{equation}
The inequality~(\ref{eq:4.43}) is derived from~(\ref{eq:4.42}). For every $\F_{q^m}$-linear $\mC,$ due to "bad" case  $x$ such that $|B_R(x,\tau n)\cap\mC|\geq L,$ there are $q^{Rmn}$ such "bad" $x.$

%The inequality~(\ref{eq:4.3}) is derived from~(\ref{eq:4.2}). For every $\F_{q}$-linear code $\mC$, due to "bad" case $X$ ($|B_R(X,\tau n)\cap\mC|\geq L$), there are $q^{Rmn}$ such ''bad'' $X.$
Since $\mC$ is linear, we have 
\begin{align*}
\triangle=&Pr_{{\mC\in\mC_{Rn}}, {x\in\F_{q^m}^n}}[|B_R(x,\tau n)\cap\mC|\geq L]\\
=&Pr_{{\mC\in\mC_{Rn}}, {x\in\F_{q^m}^n}}[|B_R(0,\tau n)\cap(\mC+x)|\geq L]\\
\leq&Pr_{{\mC\in\mC_{Rn}}, {x\in\F_{q^m}^n}}[|B_R(0,\tau n)\cap \rm Span_{\F_{q^m}}(\mC+x)|\geq L]\\
\leq&Pr_{{\mC^*\in\mC^*_{Rn+1}}}[|B_R(0,\tau n)\cap \mC^*|\geq L],
\end{align*}
where $\mC^*$ is a random $Rn+1$ dimension $\F_{q^m}$-subspace of $\F_{q^m}^n$ containing $\mC$ ( If $x\notin \mC,$then $\mC^*=\rm Span_{\F_{q^m}}(\mC, x);$ otherwise $\mC^*=\rm Span_{\F_{q^m}}(\mC, y),$ where $y$ is picked randomly from $\F_{q^m}^n\backslash \mC$).

For each integer $\ell$, $\ell\in[log_{q^m} L,L].$ Let $\mathcal{F}_\ell$ be the set of all tuples $(x_1, \cdots, x_\ell)\in B_R(0,\tau n)^\ell$ such that $x_1, \cdots, x_\ell$ are linearly independent and 
\begin{equation*}
|\rm span(x_1, \cdots, x_\ell)\cap B_R(0,\tau n)|\geq L.
\end{equation*} 
Hence, 
\begin{eqnarray*}
|\mathcal{F}_t|\leq |B_R(0,\tau n)^t|\leq\left(4q^{mn(\tau+\tau b-\tau^2b)}\right)^\ell.
\end{eqnarray*}

Let $\mathcal{F}=\bigcup_{\ell=\lceil{loq_{q^m}L} \rceil}^L \mathcal{F}_\ell.$ For each $x=(x_1,\cdots, x_\ell)\in\mathcal{F},$ let $\{x\}$ denote the set $\{x_1,\cdots,x_\ell\}.$

We claim that if $|B_R(0,\tau n)\cap \mC^*|\geq L,$ there must exist $x\in\mathcal{F}$ such that $\{x\}\subseteq \mC^*.$ %Indeed, let $\{w\}$ be a maximal linearly independent subset of $B_R(0,\tau n)\cap \mC^*.$ 
Thus, we have
\begin{align*}
\triangle\leq&Pr_{{\mC^*\in\mC^*_{Rn+1}}}[|B_R(0,\tau n)\cap \mC^*|\geq L]\\%
\leq&\sum_{x\in\mathcal{F}_\ell}Pr_{{\mC^*\in\mC^*_{Rn+1}}}[\{x\}\subseteq\mC^*]\\%
=&\sum_{\ell=\lceil{loq_{q^m}L} \rceil}^L\sum_{x\in\mathcal{F}_\ell}Pr_{{\mC^*\in\mC^*_{Rn+1}}}[\{x\}\subseteq\mC^*]\\%
=&\sum_{\ell=\lceil{loq_{q^m}L} \rceil}^L|\mathcal{F}_\ell|Pr_{{\mC^*\in\mC^*_{Rn+1}}}[\{v\}\subseteq\mC^*]\\%
%\leq&\sum_{\ell=\lceil{loq_{q^m}L} \rceil}^L \left(4q^{mn(\tau+\tau \rho-\tau^2\rho)}\right)^\ell Pr_{{\mC^*\in\mC^*_{Rn+1}}}[\{x\}\subseteq\mC^*]\\%
%\leq&\sum_{\ell=\lceil{loq_{q^m}L} \rceil}^L \left(4q^{mn(\tau+\tau \rho-\tau^2\rho)}\right)^\ell q^{m((Rn+1)+\ell-n-2)\ell+4(Rn+1)-1)}\\%
%=&4^\ell\cdot\sum_{\ell=\lceil{loq_{q^m}L} \rceil}^L q^{mn\ell(\tau+\tau \rho-\tau^2 \rho+R-1+\frac{\ell}{n}+\frac{4R}{\ell}+\frac{3}{n\ell})}\\
%\leq&4^\ell\cdot\sum_{\ell=\lceil{loq_{q^m}L} \rceil}^L q^{mn\ell(\tau+\tau b-\tau^2 \rho+R-1)}\\
%\leq&q^{-2mn}.
\end{align*}
By taking $R=(1-\tau)(1-\rho\tau)-\epsilon$, we can obtain
\begin{align*}
\triangle\leq&Pr_{{\mC^*\in\mC^*_{Rn+1}}}[|B_R(0,\tau n)\cap \mC^*|\geq L]\\%
%\leq&\sum_{x\in\mathcal{F}_\ell}Pr_{{\mC^*\in\mC^*_{Rn+1}}}[\{x\}\subseteq\mC^*]\\%
%=&\sum_{\ell=\lceil{loq_{q^m}L} \rceil}^L\sum_{x\in\mathcal{F}_\ell}Pr_{{\mC^*\in\mC^*_{Rn+1}}}[\{x\}\subseteq\mC^*]\\%
%=&\sum_{\ell=\lceil{loq_{q^m}L} \rceil}^L|\mathcal{F}_\ell|Pr_{{\mC^*\in\mC^*_{Rn+1}}}[\{v\}\subseteq\mC^*]\\%
\leq&\sum_{\ell=\lceil{loq_{q^m}L} \rceil}^L \left(4q^{mn(\tau+\tau \rho-\tau^2\rho)}\right)^\ell Pr_{{\mC^*\in\mC^*_{Rn+1}}}[\{x\}\subseteq\mC^*]\\%
\leq&\sum_{\ell=\lceil{loq_{q^m}L} \rceil}^L \left(4q^{mn(\tau+\tau \rho-\tau^2\rho)}\right)^\ell q^{m((Rn+1)+\ell-n-2)\ell+4(Rn+1)-1)}\\%
%=&4^\ell\cdot\sum_{\ell=\lceil{loq_{q^m}L} \rceil}^L q^{mn\ell(\tau+\tau \rho-\tau^2 \rho+R-1+\frac{\ell}{n}+\frac{4R}{\ell}+\frac{3}{n\ell})}\\
\leq&4^\ell\cdot\sum_{\ell=\lceil{loq_{q^m}L} \rceil}^L q^{mn\ell(\tau+\tau b-\tau^2 \rho+R-1)}\\
\leq&q^{-2mn}.
\end{align*} 

Thus, an $\F_{q^m}$-linear self-orthogonal rank metric code with rate  $R=(1-\tau)(1-\rho\tau)-\epsilon$ is not ${(\tau n, \exp(O_{\tau, q}(\frac{1}{\epsilon})))}$-list decodable with an exponential small probability $q^{-2mn}.$

\end{proof}
%%%%%%%%%%%%%%%%%%%%%%%%%%%%%%%%%%%%%%%%%%%%%%
\section{Conclusion}
We investigate the list decodable $\F_q$ and $\F_{q^m}$-linear self-orthogonal rank metric codes. We show that the list decodability of $\F_q$-linear self-orthogonal rank metric codes is as good as that of general random rank metric codes as well, which can be list decoded up to the Gilbert-Varshamov bound. By using the same methods for $\mathbb{F}_{q^m}$-linear rank metric codes, our results reveal that the $\F_{q^m}$-linear self-orthogonal rank metric codes is list decodable up to $R=(1-\tau)(1-\rho\tau)-\epsilon$ with exponential list size. The list size of the codes grows polynomially in $q^m$ (rather than just $q$). It is interesting to decrease the list size of $\F_{q^m}$-linear self-orthogonal rank metric codes and $\F_{q^m}$-linear rank metric codes.

\end{document}